\documentclass[letterpaper, 10 pt, conference]{ieeeconf}

\IEEEoverridecommandlockouts
\let\labelindent\relax

\usepackage{color}
\usepackage{enumerate,graphicx,epstopdf}
\usepackage{amsmath,amssymb,bm,amsthm,nicefrac}
\usepackage{cite}
\usepackage{mathtools}
\usepackage[shortlabels]{enumitem}
\usepackage{array}
\usepackage{algpseudocode,algorithm,algorithmicx}
\usepackage{color}
\usepackage{booktabs}
\usepackage[hyphens]{url}
\usepackage{lipsum}
\usepackage[dvipsnames]{xcolor}
\usepackage[deletedmarkup=sout,authormarkup=superscript]{changes}
\usepackage[normalem]{ulem}

\usepackage{mwe}
\newtheorem{ass}{Assumption}

\newtheorem{thm}{Theorem}

\newcommand{\reals}{\mathbb{R}}

\definechangesauthor[name={Marco}, color=NavyBlue]{MN}
\definechangesauthor[name={Yaashia}, color=Plum]{YG}

\usepackage{enumitem}
\newlist{steps}{enumerate}{1}
\setlist[steps, 1]{label = Step \arabic*:}
\begin{document}
\title{Stability Analysis of Hypersampled Model Predictive Control 
\thanks{The authors are with the University of Colorado Boulder.}
\thanks{This research is supported by the NSF-CMMI Award  2046212.}}

\author{Yaashia Gautam, Marco M. Nicotra}

\maketitle

\begin{abstract}
This paper introduces a new framework for analyzing the stability of discrete-time model predictive controllers acting on continuous-time systems. The proposed framework introduces the distinction between \emph{discretization time} (used to generate the optimal control problem) and \emph{sampling time} (used to implement the controller). The paper not only shows that these two time constants are independent, but also motivates the benefits of selecting a sampling time that is smaller than the discretization time. The resulting approach, hereafter referred to as Hypersampled Model Predictive Control, overcomes the traditional trade-off between performance and computational complexity that arises when selecting the sampling time of traditional discrete-time model predictive controllers.

\end{abstract}

\section{Introduction}
Model Predictive Control (MPC) is a popular 
constrained control strategy for nonlinear systems. The idea behind MPC is to solve an Optimal Control Problem (OCP) at every instant and apply the first step of the optimal control sequence to the system. Although MPC is widely used due to its stability, feasibility, and robustness properties \cite{mayne2000mpc}, its widespread adoption is hindered by the fact that solving the optimal control problem in real time can be challenging. To address this issue, extensive research has been dedicated to the development of increasingly efficient numerical solvers \cite{tenny2004nonlinear,diehl2009efficient,quirynen2015autogenerating,fbstab}. This paper investigates a practical stratagem that reduces computational complexity by decoupling the prediction model (used to formulate the OCP) from the system model (used to prove closed-loop stability). \smallskip

Given a continuous-time system, MPC can be formally implemented in one of three ways: the most popular approach \cite{MPC4} is to discretize the dynamics and use discrete-time MPC (DT-MPC). Another option is sampled-data MPC \cite{magni2004model}, where the OCP features piecewise-constant control inputs and continuous-time dynamics. In both cases, the prediction model is built using the same sampling rate of the controller. Thus, increasing the sampling rate also increases the numerical complexity of the OCP. The final option is Dynamically Embedded MPC (DE-MPC), which replaces the OCP solver with a dynamic compensator that runs parallel to the system \cite{feller2017,DE-MPC}. In this case, the controller is continuous whereas the prediction model is discrete.\smallskip 

This paper provides formal stability proofs for a fourth strategy, hereafter denoted Hypersampled MPC (HMPC), which decouples the sampling rate of the controller from the timestep used to perform trajectory predictions.  
This stratagem, which some practitioners employ despite the absence of theoretical guarantees, reduces the computational burden of MPC by solving a coarse OCP (thereby reducing computational effort) and running the controller at a high sample rate (thereby improving reaction times).\smallskip

The paper is organized as follows. Section III summarizes
the ideal continuous-time MPC formulation
and details its main implementation challenges. Section IV analyzes how the discretization of
the continuous-time OCP affects the stability of the 
closed-loop system.
Section V addresses how the stability of the closed-loop system is affected by the sampling time.
Section VI features numerical comparisons between the proposed method and two traditional discrete-time MPC schemes: one implemented at a high sampling rate and one implemented at a low sampling rate.

\section{Notation}
The paper employs the following norm conventions. $\mathbb{R}^n$ is the set of $n-$dimensional vectors of real numbers, $\mathbb{N}$ is the set of natural numbers, and $\mathbb N^+$ is the set of non-zero natural numbers. Given a vector $x\in\mathbb R^n$, its 2-norm is $\|x\|\coloneqq\sqrt{x^\top x}$. Given a positive semi-definite matrix $Q\geq0$, the matrix norm of vector $x$ is $\|x\|_Q\coloneqq\sqrt{x^\top Qx}$. Given a function $x(t)$, its limit superior satisfies
\begin{equation} \exists~t^\star\geq0:~~\overline{\lim_{t\to\infty}}\|x(t)\|\geq\|x(\tau)\|,~~\forall \tau\geq t^\star,
\end{equation}
and its infinity norm is $\|x(t)\|_\infty\coloneqq\sup_{t\in[0,\infty)} \|x(t)\|$.

\section{Problem Statement}

Consider a continuous-time system
\begin{equation} \label{eq:continuous sys}
    \dot x(t) = f(x(t),u(t))
\end{equation}
where $x \in \mathbb{R}^{n}$ is the state, $u \in \mathbb{R}^m$ is the control input, and $f : \mathbb{R}^{n} \times \mathbb{R}^m \to \mathbb{R}^{n}$ is the system dynamics. The system is subject to state and input constraints $x \in \mathcal{X}$ and $u\in \mathcal{U}$, where $\mathcal{X}\subseteq\mathbb R^n$ and $\mathcal{U}\subseteq\mathbb R^m$ are closed convex sets.
\smallskip

This system can be controlled using continuous-time MPC (CT-MPC) by solving the optimal control problem
\begin{subequations}\label{ct_ocp}
\begin{align}
\min_{\mu(\tau)} ~~& J(\xi(T))+\int_0^T l(\xi(\tau),\mu(\tau))d\tau\\
\mathrm{s.t.} ~~&\xi(0)=x,\\
&\dot \xi(\tau)= f(\xi    (\tau),\mu(\tau))~~\quad\!\!\forall \tau\in[0,T],\\
        &\mu(\tau) \in \mathcal{U},\qquad\qquad\qquad\:\!\!\forall \tau\in[0,T],\\
        &\xi(\tau) \in \mathcal{X},\qquad\qquad\qquad\:\!\!\forall \tau\in[0,T],\\
        &\xi(T) \in \Omega,
\end{align}
\end{subequations}
where $\xi: [0,T] \to \reals^n,\mu : [0,T] \to \reals^m$  are the predicted state and input, $T>0$ is the prediction horizon, $J:\mathbb{R}^n\to\mathbb{R}$ is the terminal cost, $l:\mathbb{R}^n\times\mathbb{R}^m\to\mathbb{R}$ is the incremental cost and $\Omega \subseteq\mathcal X$ is the terminal constraint set. Specifically, the CT-MPC feedback law is $u(x)=\mu^*(0|x)$, where $\mu^*(\tau|x)$ denotes the solution mapping of \eqref{ct_ocp} for a given $x$. As detailed in \cite{mayne2000mpc}, the following assumption is sufficient to ensure that the optimal control problem \eqref{ct_ocp} is recursively feasible and that solving it leads to a stabilizing control action. 

\begin{ass}\label{ct_ocp assumption}
The dynamics $f : \mathbb{R}^{n} \times \mathbb{R}^m \to \mathbb{R}^{n}$ are Lipschitz continuous, stabilizable, and admit the origin as a constraint-admissible equilibrium point.  
The cost functions $J:\mathbb{R}^n\to\mathbb{R}$ and $l:\mathbb{R}^n\times\mathbb{R}^m\to\mathbb{R}$ are convex, twice continuously differentiable, zero at the origin, and upper/lower bounded by quadratic functions. In addition, $\Omega\subseteq\mathcal X$ is a closed and convex set such that there exists a terminal control law $\kappa:\mathbb{R}^n \to \mathbb{R}^{m}$ that satisfies
\begin{subequations}
\begin{align}
\kappa(x) &\in \mathcal{U},&\forall x\in\Omega,\\
f(x,\kappa(x)) &\in T_{\Omega}(x),&\forall x\in\Omega,\\
\nabla J(x)+l(x,\kappa(x)) &\leq 0,&\forall x\in\Omega,
\end{align}
\end{subequations}
where $T_{\Omega}(x)$ is the tangent cone (defined in \cite{BLANCHINI19991747}) of $\Omega$ evaluated in $x$.
\end{ass}

Under these standard assumptions, the following theorem is easily deduced from existing literature \cite{mayne2000mpc}.
\begin{thm}\label{thm1}
Under Assumption \ref{ct_ocp assumption}, there exists a class $\mathcal K$ function $\gamma_1$ and positive scalars $\chi_1,\Delta_1$ such that
\begin{equation}\label{eq:Ideal_CL}
    \dot x=f(x,\mu^*(0|x)+\delta_1)
\end{equation}
satisfies
\begin{equation}
    \overline{\lim_{t\to\infty}}\|x(t)\|\leq \gamma_1 \left( \ \overline{\lim_{t\to\infty}}\|\delta_1(t)\|\right),
\end{equation}
whenever $\|x(0)\| \leq \chi_1$ and $\|\delta_1\|_\infty \leq\Delta_1$.
\end{thm}
\begin{proof} As detailed in
\cite[Section 3.6]{mayne2000mpc}, the origin of \eqref{eq:Ideal_CL} is exponentially stable for $\delta_1=0$. Therefore, it is locally Input-to-State Stable (ISS) for sufficiently small $\delta_1\neq0$.
\end{proof}

This result states that CT-MPC is robust with respect to additive input disturbances. Unfortunately, its implementation 
poses two major challenges: i) the OCP \eqref{ct_ocp}  is formulated in a Banach function space and requires sophisticated solvers, e.g. \cite{HAUSER2002377}, to compute $\mu^*(\tau|x)$. ii) Implementing the controller in continuous-time requires solving \eqref{ct_ocp} ``instantaneously'', which is not realistic in practice.\smallskip

Due to the limitations of CT-MPC, it is common practice to discretize the system dynamics and employ DT-MPC. This approach assumes that the dynamic model of the system matches the prediction model of the OCP. An unfortunate consequence is that, for a fixed prediction horizon, reducing the sampling time inevitably leads to an increased number of prediction steps. This has the combined negative effect of increasing the numerical complexity of the OCP while also decreasing the allocated time for solving it. 

\smallskip

This paper seeks to formalize the distinction between \emph{discretization time} $t_d$, i.e., the step size used to discretize the dynamic model \eqref{eq:continuous sys}, and \emph{sampling time} $t_s$, i.e., the time at which the controlled is implemented. Doing so will require substantially different stability proofs from traditional MPC literature (e.g. \cite{mayne2000mpc,MPC4,magni2004model}), which heavily relies on having a prediction model that matches the system dynamics.

\section{Discretized Optimal Control Problem}
This section addresses the first challenge of CT-MPC, which is that \eqref{ct_ocp} is an infinite dimensional OCP. Given $N \in \mathbb{N}^+$, let $t_d = \nicefrac{T}{\!\,}_N$ be the discretization time and let $f_d : \mathbb{R}^{n} \times \mathbb{R}^m \to \mathbb{R}^{n}$ be a discrete-time model such that
\begin{equation}
\lim_{t_d\to0}\frac{f_d(x,u)-x}{t_d}=f(x,u).
\end{equation}
Then, the continuous-time optimal control problem \eqref{ct_ocp} can be approximated as
\begin{subequations}\label{dt_ocp}
\begin{align}
\min_{\mu_j} ~~& J(\xi_N)+\sum_{j=0}^{N-1}l_d(\xi_j,\mu_j)\\
\mathrm{s.t.} ~~&\xi_0=x,\\
&\xi_{j+1}= f_d(\xi_j,\mu_j) ~~\quad\forall j\in[0,N-1],\\
        &\mu_j \in \mathcal{U},\qquad\qquad\qquad\!\!\forall j\in[0,N-1],\\
        &\xi_j \in \mathcal{X},\qquad\qquad\qquad\!\!\forall j\in[1,N-1],\\
        &\xi_N \in \Omega,
\end{align}
\end{subequations}
with $l_d(x,u)=t_d\,l(x,u)$.  Let $\mu^*_j(x)$ denote the solution mapping of \eqref{dt_ocp}. The following assumption ensures that $\mu^*_j(x)$ is a suitable approximation of $\mu^*(jt_d\,|x)$. 
\begin{ass}
\label{ass:regular}
The parametrized optimal control problem \eqref{ct_ocp} is strongly regular. Moreover, its discrete approximation \eqref{dt_ocp} is strongly regular and admits $T_d>0$ and $\chi_0>0$ such that the discretization error $\Delta \mu=\mu_0^*(x)-\mu^*(0|x)$ satisfies
\begin{equation}\label{lemma_result}
    \|\Delta \mu\| \leq L(t_d)\|x\|, \qquad \forall \ \|x\|\leq \chi_0,
\end{equation}
where $L(t_d)>0,~\forall t_d\in(0,T_d]$ is a Lipschitz constant,  and 
\begin{equation}
\lim_{t_d\to0}L(t_d)=0.
\end{equation}
\end{ass}

Although Assumption \ref{ass:regular} is formally difficult to verify (see \cite{dontchev2019lipschitz} for the case $f_d(x,u)=x+t_df(x,u)$ and  $\mathcal{X},\Omega=\reals^n$), it is reasonable to expect that \eqref{ct_ocp} and \eqref{dt_ocp} converge to the same solution as the discretization step $t_d$ goes to zero. The following theorem states that, given a sufficiently small $t_d$, the CT-MPC law $u(x)=\mu^*(0|x)$, which relies on the solution to  \eqref{ct_ocp}, can be replaced by the approximate control law $u(x)=\mu_0^*(x)$, obtained by solving \eqref{dt_ocp}.

\begin{thm}\label{thm2}
Under Assumptions \ref{ct_ocp assumption}-\ref{ass:regular}, given a sufficiently small discretization step $t_d >0$, there exist positive scalars $\chi_2,\Delta_2$ such that the origin of the closed-loop system
\begin{equation}\label{eq:MPC_DTOCP}
    \dot x = f(x,\mu^*_0(x)+\delta_2)
\end{equation}
is locally ISS with state and input restrictions $\|x(0)\|\leq \chi_2 $ and $\|\delta_2\|_\infty\leq\Delta_2$.
\end{thm}

\begin{proof}
System \eqref{eq:MPC_DTOCP} can be rewritten as 
\begin{equation}\label{eq:thm2}
    \dot x = f(x,\mu^*(0|x)+\Delta \mu+\delta_2).
\end{equation}
Given $\delta_2=0$, it follows from Theorem \ref{thm1} that system \eqref{eq:thm2} is locally ISS with asymptotic gain
\begin{equation}
\overline{\lim_{t\to\infty}}\|x(t)\|\leq \gamma_1 \left( \ \overline{\lim_{t\to\infty}}\|\Delta \mu(t)\|\right).
\end{equation}
Moreover, it follows from Assumption \ref{ass:regular} that the solution error between \eqref{ct_ocp} and \eqref{dt_ocp} satisfies
\begin{equation}\overline{\lim_{t\to\infty}}\|\Delta \mu(t)\|\leq L(t_d) \:\overline{\lim_{t\to\infty}}\|x(t)\|.
\end{equation}
Referring to Figure \ref{interconnect1}, it follows from the small gain theorem that the interconnection between the dynamic system \eqref{eq:thm2} and the static system $\Delta \mu=\mu_0^*(x)-\mu^*(0|x)$, is locally ISS with respect to $\delta_2\neq0$ if the  condition
\begin{equation}\label{eq:small_gain}
    L(t_d)\gamma_1(s)<s
\end{equation}
is satisfied for all $s\leq\Delta_1$. Since Assumption \ref{ass:regular} guarantees $\lim_{t_d\to0}L(t_d)=0$, it is always possible to select a sufficiently small $t_d$ such that \eqref{eq:small_gain} holds. Given an input restriction $\Delta_2\in(0,\Delta_1)$, it is then possible to select a state restriction $\chi_2>0$ such that $\|x(0)\|\leq\chi_2$ ensures $\|x(t)\|\leq\min(\chi_0,\chi_1)$ and $\|\Delta\mu(t)\|\leq\Delta_1-\Delta_2$, $\forall t\geq0$.
\end{proof}

Theorem \ref{thm2} states that it is possible to implement an approximate CT-MPC law by solving \eqref{dt_ocp} instead of \eqref{ct_ocp}. Although \eqref{dt_ocp} is a standard Non-Linear Program (NLP), which is arguably simpler to solve compared to \eqref{ct_ocp}, it would  be unrealistic to assume that the NLP solution can be computed ``instantaneously''.

\section{Hypersampled Model Predictive Control}
This section addresses the second challenge of CT-MPC, which is the fact that solving \eqref{dt_ocp} requires a finite amount of time. 
To address this issue, let $t_s>0$ be the sampling time of the control law, let 
\begin{equation}\label{eq:sampled_state}
x_k(t)=x(kt_s),\qquad\forall t\in[kt_s,(k\!+\!1)t_s). 
\end{equation}
with $k\in\mathbb N$, be the sampled state of the system, and let $u(t)=\mu_0^*(x_k(t))$ be a zero-order hold control input.

 The following theorem states that, given a sufficiently small sampling time, the closed-loop system maintains all of its local ISS properties.
 \begin{figure}
 \vspace{10pt}
    \centering
    \includegraphics[trim=5cm 8cm 10cm 4cm, clip=true, scale=0.5]{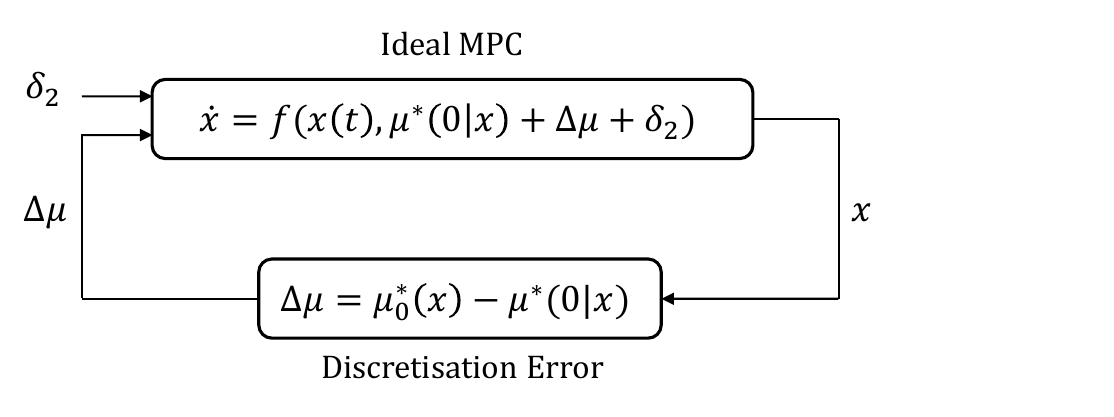}
    \caption{Discretized Model Predictive Control law applied to the continuous system. The interconnection is ISS with respect to the error introduced by discretization.
    \label{interconnect1}}
\end{figure}
\vspace{-5pt}
\begin{figure*}
\vspace{10pt}
    \centering
    \includegraphics[trim=1cm 5cm 7cm 3.5cm, clip=true, scale=0.47]{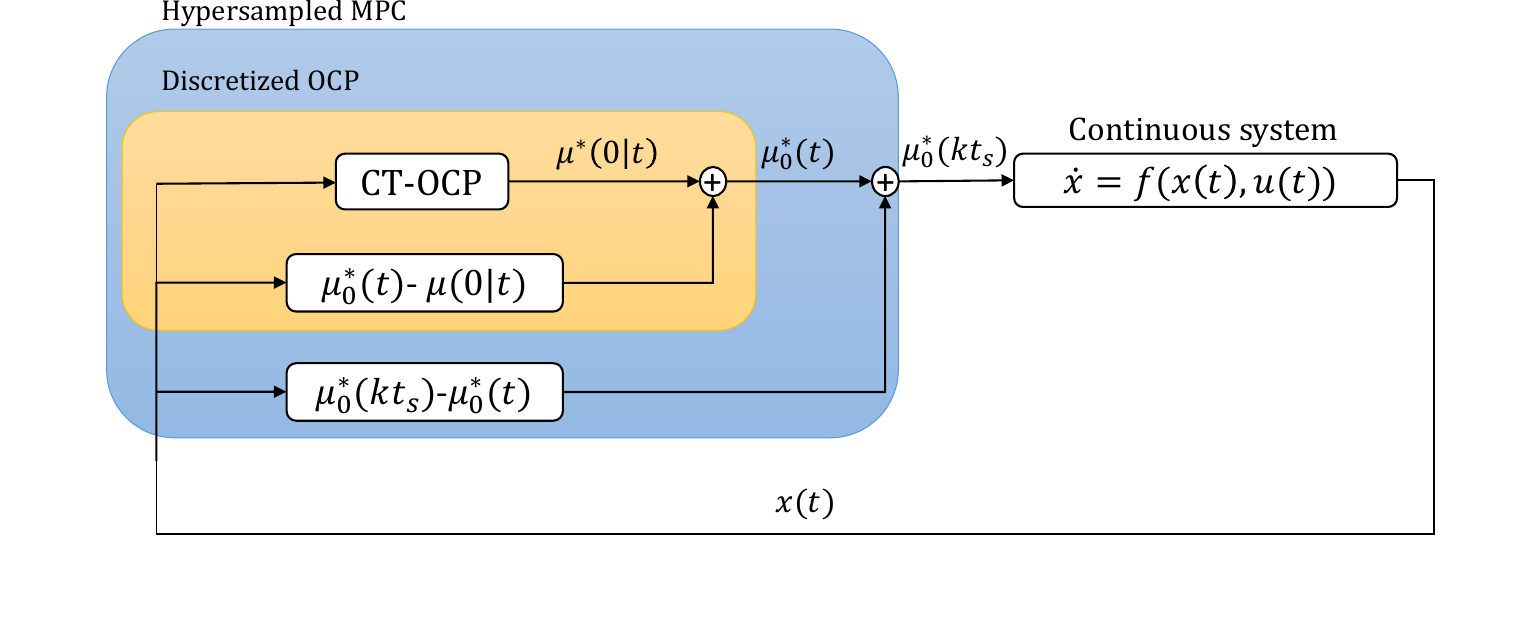}
    \caption{Hypersampled Model Predictive Control. The continuous system is sampled at timestep $t_s$ to generate $x_k$, which is then used to compute the optimal control law to the discretized MPC problem. \label{interconnect2}}

\end{figure*}
\smallskip

\begin{thm}\label{thm3}
Under Assumptions \ref{ct_ocp assumption}-\ref{ass:regular}, given a sufficiently small discretization step $t_d>0$ and sampling time $t_s>0$, there exist positive scalars $\chi_3,\Delta_3$ such that the origin of the closed-loop system
\begin{equation}\label{eq:HMPC_cl}
    \dot x = f(x,\mu^*_0(x_k)+\delta_3)
\end{equation}
is locally ISS with state and input restrictions $\|x(0)\|\leq \chi_3 $ and $\|\delta_3\|_\infty\leq\Delta_3$.
\end{thm}
\begin{proof}
    Due to Theorem \ref{thm2}, there exists a sufficiently small $t_d>0$ such that \eqref{eq:MPC_DTOCP} is locally ISS with respect to input disturbances. It then follows from \cite [Theorem 2]{teel} that there exists a sufficiently small $t_s>0$ such that the sampled-data implementation \eqref{eq:HMPC_cl} is also locally ISS.
\end{proof}

Figure \ref{interconnect2} clarifies the various components of HMPC. Theorem \ref{thm3} highlights the fact that the discretization time $t_d$ and the sampling time $t_s$ can be treated as two entirely different entities. This distinction introduces an additional degree of freedom that can be leveraged for HMPC design. In particular, we note the following properties: 
\begin{itemize}
    \item \textbf{Benefits of Decreasing $\bm{t_s}$:} Reducing the sampling time tends to improve the performance of the controller by making it more reactive to external disturbances. Generally speaking, one would like $t_s$ to be as small as possible to mimic continuous-time behavior.
    \item \textbf{Limit for Decreasing $\bm{t_s}$:} Due to real-time implementation requirements, $t_s$ is lower-bounded by the computational time required to solve \eqref{dt_ocp}. 
    \item \textbf{Benefits of Increasing $\bm{t_d}$:} Given a fixed prediction horizon $T$, a larger discretization time step $t_d$ leads to less prediction steps $N$. Since the computational complexity of  \eqref{dt_ocp} scales with $N$, it is generally beneficial for $t_d$ to be as large as possible.
    \item \textbf{Limit for Increasing $\bm{t_d}$:} Since larger values of $t_d$ increase the discrepancy between $\mu^*(0|x)$ and $\mu_0^*(x)$, an upper bound on the maximal admissible error will translate into an upper bound on $t_d$.
\end{itemize}
Since traditional MPC inherently assumes $t_s=t_d$, it can be challenging to reduce the sampling time because doing so increases the computational complexity while tightening the real-time requirements. The HMPC framework overcomes this preconceived trade-off by decoupling the two effects: given a fixed discretization time $t_d$, the OCP complexity is unaffected by the sampling time $t_s$.\smallskip

In practice, any discrete-time MPC algorithm can be implemented as HMPC by selecting a coarse discretization step $t_d$ and then running the controller at a faster sampling rate $t_s$. As a result, the HMPC framework can be adopted without any additional effort from the user. Moreover, since Theorem \ref{thm3} proves that the closed-loop system is locally ISS with respect to an input disturbances $\delta_3$, the proposed scheme can be seamlessly combined with other MPC schemes that rely on ISS-based proofs. Notably, HMPC can be combined with DE-MPC \cite{DE-MPC} by replacing the numerical solver used to compute $\mu_0^*(x)$ with a dynamic compensator that tracks the solution of \eqref{dt_ocp} with a bounded error.

\section{Numerical Examples}
This section illustrates the benefits of distinguishing between $t_d$ and $t_s$ using a linear and a nonlinear example. All timing data is obtained using FBstab \cite{fbstab} to solve the OCP.

\subsection{Double Integrator}

Consider the double integrator system 
\begin{equation}\label{integrator}
    \dot x = \begin{bmatrix}
    0&1\\0&0
    \end{bmatrix}x+\begin{bmatrix}0\\1
    \end{bmatrix}u+\begin{bmatrix}0\\1
    \end{bmatrix}d.
\end{equation}
where $d$ is an external disturbance and the state and input constraint sets are $\mathcal{X} = \{x \in \reals^2 : \|x_1\| \leq 2\ , \|x_2\| \leq 0.4\}$ and $\mathcal{U}= \{u \in \reals: -4 \leq u \leq 10 \}$. Given the initial condition $x_0=[2,0]^\top$ and control horizon $T=2$, consider the OCP
\begin{subequations}\label{OCP}
\begin{align}
        \min ~~& \|\xi(T)\|^2_P+\int_0^T \|\xi(\tau)\|^2_Q+\|\mu(\tau)\|^2_R~d\tau \\
    \mathrm{s.t.}~~&\xi(0)=x,\\&  \dot \xi=A\xi+B\mu,\qquad\quad\forall \tau\in[0,T]\\ 
    &\mu(\tau) \in \mathcal{U},\qquad\qquad\quad\forall \tau\in[0,T],\\
        &\xi(\tau) \in \mathcal{X},\qquad\qquad\quad\forall \tau\in[0,T],\\
        &\xi(T) \in \mathcal{O}_\infty,
\end{align}
\end{subequations}
where $Q=\mathrm{diag}(1,0)$, $R=0.04$, $P$ is the solution to the algebraic Riccati equation $A^\top P\!+\!PA\!-\!PBR^{-1}B^\top P\!+\!Q\!=\!0$, and $\mathcal{O}_\infty$ is the maximal output admissible set \cite{gilbert1991linear} associated to the corresponding linear-quadratic regulator. \\
\noindent To study the effect of discretization, consider the OCP
\begin{subequations}\label{int_dt}
\begin{align}
        \min ~~& \|\xi_N\|^2_P+\sum^{N-1}_{i=0}  \|\xi_i\|^2_{Q_d}+\|\mu_i\|^2_{R_d} \\
    \mathrm{s.t.}~~&  \xi_{i+1}=A_d\xi_i+B_d\mu_i, \qquad \forall i\in[0,N-1]\\
    &\xi_0=x,\\
    &\mu_i \in \mathcal{U},\qquad\qquad\quad\forall i\in[0,N-1],\\
        &\xi_i \in \mathcal{X},\qquad\qquad\quad\forall i\in[1,N]
\end{align}
\end{subequations}
where $A_d, B_d$ are the discretized version of the continuous dynamics \eqref{integrator} and $Q_d=t_dQ$, $R_d=t_dR$, $N=T/t_d\in\mathbb{N}^+$. 
\smallskip
\begin{figure}
    \centering
    \includegraphics[trim=2cm 1.3cm 0.75cm 0cm, clip=true, scale=0.285]{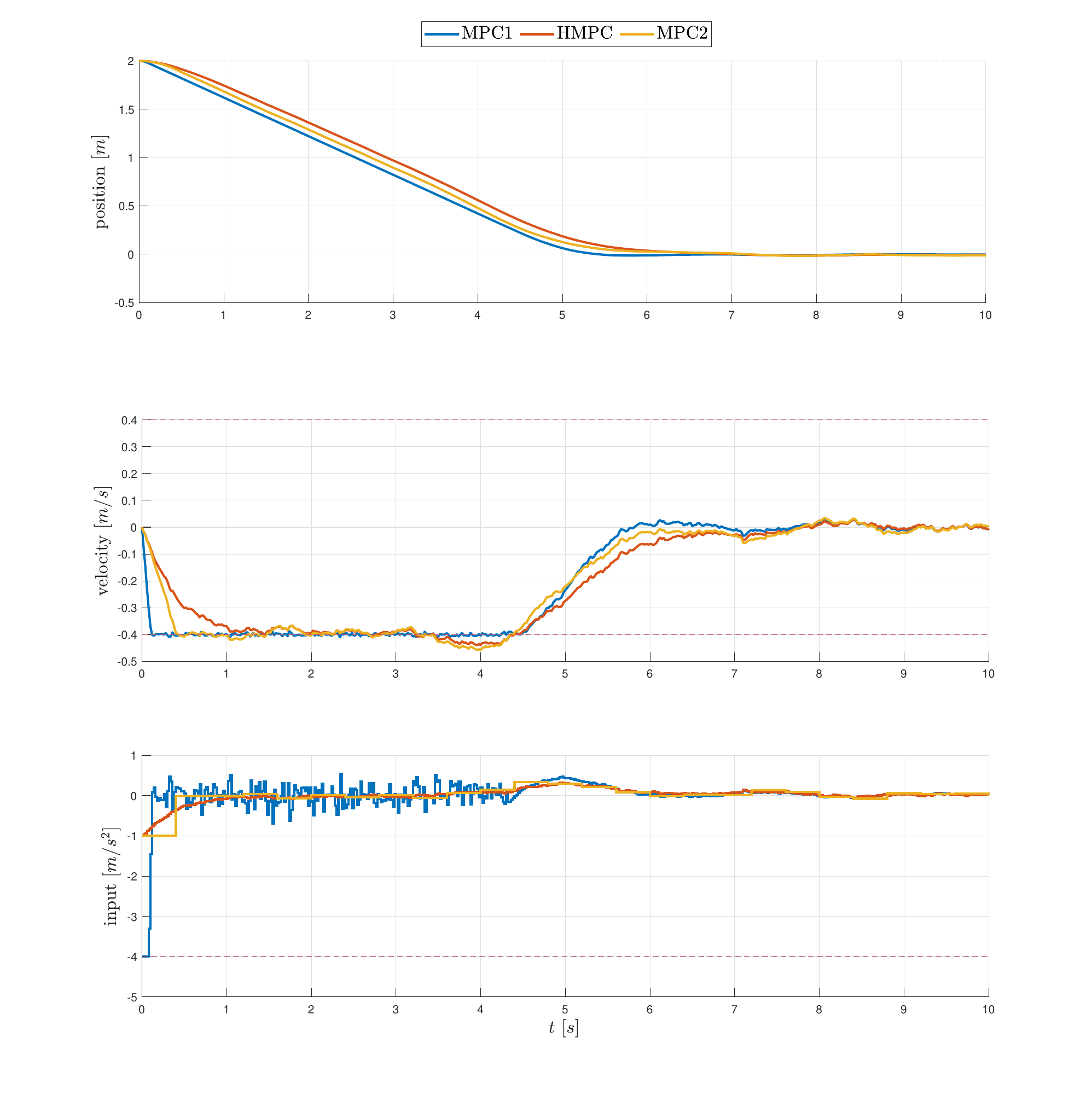}
    \caption{Comparison of the position, velocity, and input trajectories for the three MPC schemes subject to an additive disturbance on the input.\label{traj_int}}
\end{figure}
 
Figure \ref{traj_int} compares the closed loop response obtained using three different schemes: 
\begin{itemize}
    \item MPC1: $\:t_s=t_d=0.02$;
    \item HMPC: $t_s=0.02$ and $t_d=0.4$;
    \item MPC2: $\:t_s=t_d=0.4$.
\end{itemize}
Although the position trajectories for the three schemes perform similarly, the velocity and input trajectories highlight how the three schemes react to external disturbances in proximity to the constraint boundary. MPC1 is the only scheme that enforces the velocity constraints in the presence of disturbances. HMPC and  MPC2 both violate constraints, with HMPC featuring slightly better disturbance rejection properties due to its to faster response time. \smallskip 

The disadvantage of MPC1 becomes apparent by examining the computation time required to solve the underlying OCP \eqref{int_dt}. As shown in Figure \ref{time_di}, the time required to solve MPC1 is approximately 30 times more than HMPC and MPC2. This is due to the fact that MPC1 has a prediction length of $N=100$ steps, whereas HMPC and MPC2 have a prediction length of only $N=5$ steps.
\smallskip

The interest in the HMPC framework is that it enables the following design choice: ``Given a sampling time $t_s$, is it possible to reliably solve discretized OCP with $t_d=t_s$?'' If so, implementing the MPC1 solution is arguably the best option. If not, the computational cost can be reduced by maintaining the same sampling rate while selecting $t_d> t_s$ to reduce the number of optimization variables. The resulting control law is preferable than increasing the sampling time is increased to match the new $t_d$. Although the advantage of HMPC over MPC2 is marginal in the case of the double integrator, the following example shows that the difference between the two schemes can be quite severe.

\begin{figure}
\vspace{10pt}
    \centering
    \includegraphics[trim= 1cm 0cm 0cm 0cm, clip=true, scale=0.42]{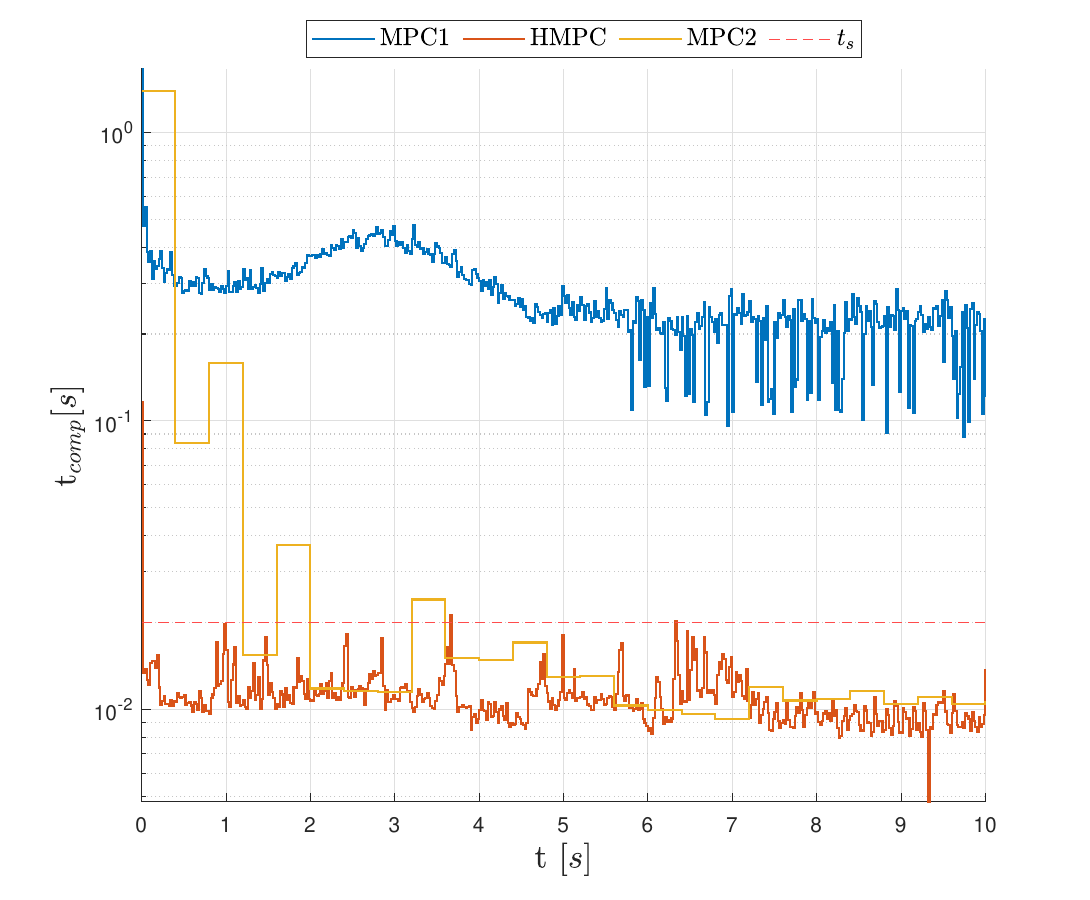}
    \caption{The computation time for the three schemes with noise, compared to the sampling time $t_s=0.02$. HMPC is te only scheme that consistently satisfies the real-time requirements. \label{time_di} }
\end{figure}

\begin{figure}
\vspace{5pt}
    \centering
    \includegraphics[trim=2cm 2.8cm 1cm 0.5cm, clip=true,scale=0.275]{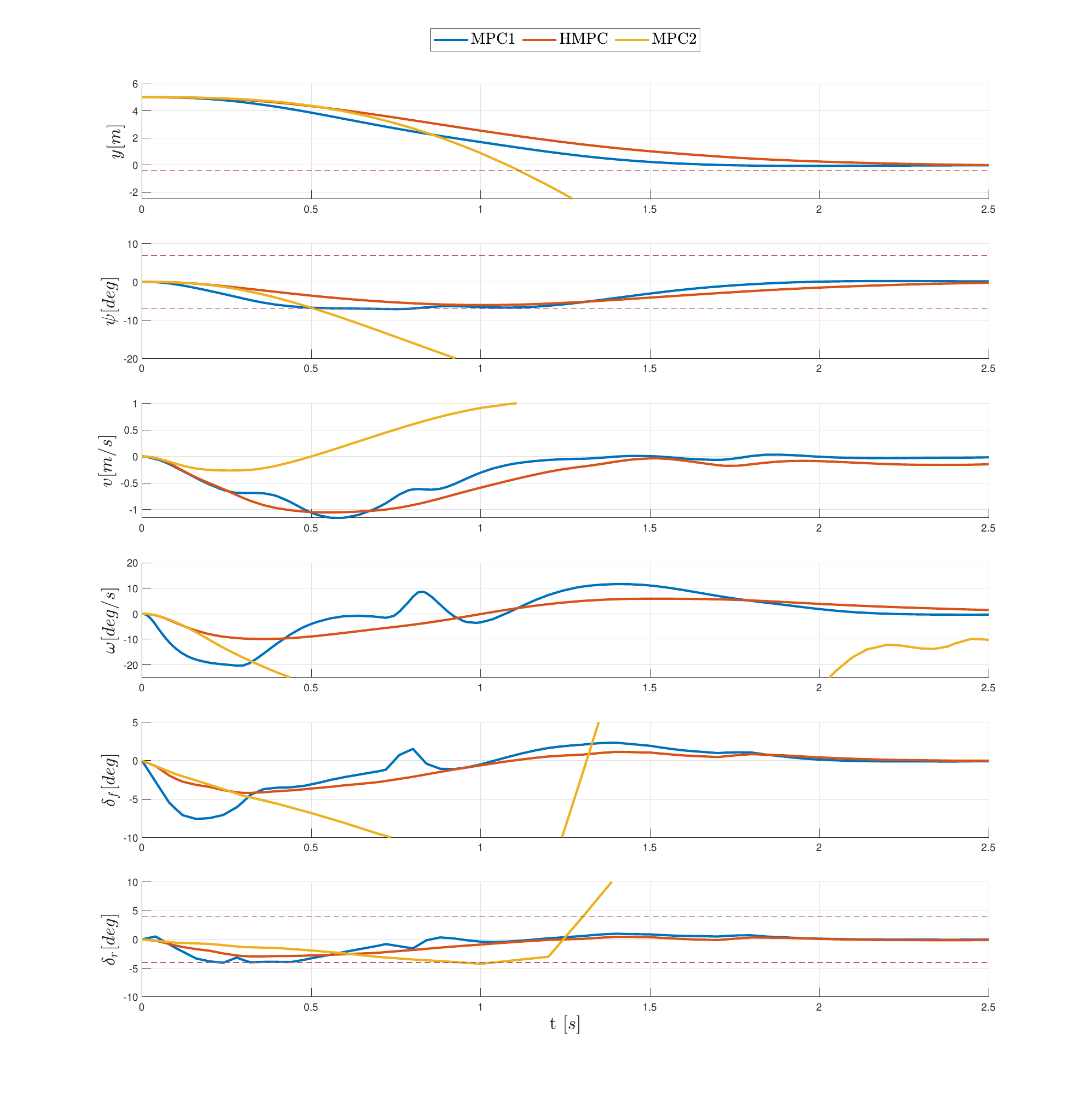}
    \caption{The state trajectories for the nonlinear lane change model. The trajectories of MPC1 and HMPC perform similarly, whereas MPC2 leads to an unstable closed-loop response.\label{nonlinear}}

\end{figure}

\subsection{Nonlinear Lane Change}
To emphasize the advantages of HMPC with a more complex example, we consider the nonlinear lane change system detailed in \cite{TDO_MPC}. The states and
control inputs are $x=[y,\psi,v,\omega,\delta_f,\delta_r]^\top$ and $u=[\dot \delta_f,\dot \delta_r]$, where, $y$ is the lateral position, $v$ is the lateral velocity, $\psi$ is the yaw angle, $\omega$ is the yaw rate, $\delta_f$ is the
front steering angle, and $\delta_r$ is the rear steering angle. The system dynamics are
\begin{subequations}\label{nonlin}
\begin{align}
\dot y &= s \ \sin(\psi) + v \ \cos(\psi)\\
\dot \psi&= \omega\\
\dot v&= -s \ \omega+\frac{F\alpha_f \cos(\delta_f)+F\alpha_r \cos(\delta_r)+F_w}{m}\\
\dot \omega&=\frac{F\alpha_f \cos(\delta_f)l_f-F\alpha_r \cos(\delta_r)l_r}{I_{zz}}\\
\dot \delta_f&= u_1, \qquad \dot \delta_r=u_2 
\end{align}
\end{subequations}
The functions $F(\alpha), \alpha_f,\alpha_r, F_w$ along with parameters such as $m,I_{zz},l_f,l_r,\mu$ can be found in \cite{TDO_MPC}. We take the initial condition $x_0=[5,0,0,0,0,0]^\top$ and use the same cost function as \eqref{OCP}, with $Q=I_6$, $R=I_2$, and $T=2$. The constraints on the system are 
\begin{subequations}
    \begin{align*}
    y\in[-0.4,10], &&\psi \in[-7^{\circ},7^{\circ}]\\
    \delta_f\in[-35^\circ,35^\circ],&&\delta_r\in[-4^\circ,4^\circ]\\
u_1\in[-1.2,1.2],&& u_2\in[-0.6, 0.6]
    \end{align*}
\end{subequations}

Figure \ref{nonlinear} and Figure \ref{nonlinear_in} compare the state and input trajectories obtained using three different schemes: 
\begin{itemize}
    \item MPC1: $\:t_s=t_d=0.04$;
    \item HMPC: $t_s=0.04$ and $t_d=0.2$;
    \item MPC2: $\:t_s=t_d=0.2$.
\end{itemize}

The trajectories are comparable for the MPC1 and HMPC schemes, whereas implementing MPC2 causes the system to become unstable. This is due to the fact that, although $0.4\mathrm s$ is too large to be a suitable sampling time, it is still sufficiently small to ensure that \eqref{dt_ocp} is a suitable approximation of \eqref{ct_ocp}. \smallskip

Figure \ref{nonlinear_time} plots the computational time taken at each step for solving the OCP for the three schemes. As with the double integrator, the time taken to solve the OCP of MPC1 is consistently larger  more than the time used for solving the OCP of the HMPC (5 times on average). 
Note that the time taken to implement MPC2 was omitted from the figure due to the fact that its unstable behavior led to significantly larger computation times. Once again, we note that MPC1 cannot be implemented in real-time because the computational time required to solve the OCP is higher than the sampling time $t_s$. Conversely, HMPC significantly reduce the computational effort by choosing a coarse discretization step $t_d$ while keeping $t_s$ unchanged.

\begin{figure}
\vspace{5pt}
    \centering
    \includegraphics[trim=2cm 2.5cm 0.75cm 0.0cm, clip=true,scale=0.265]{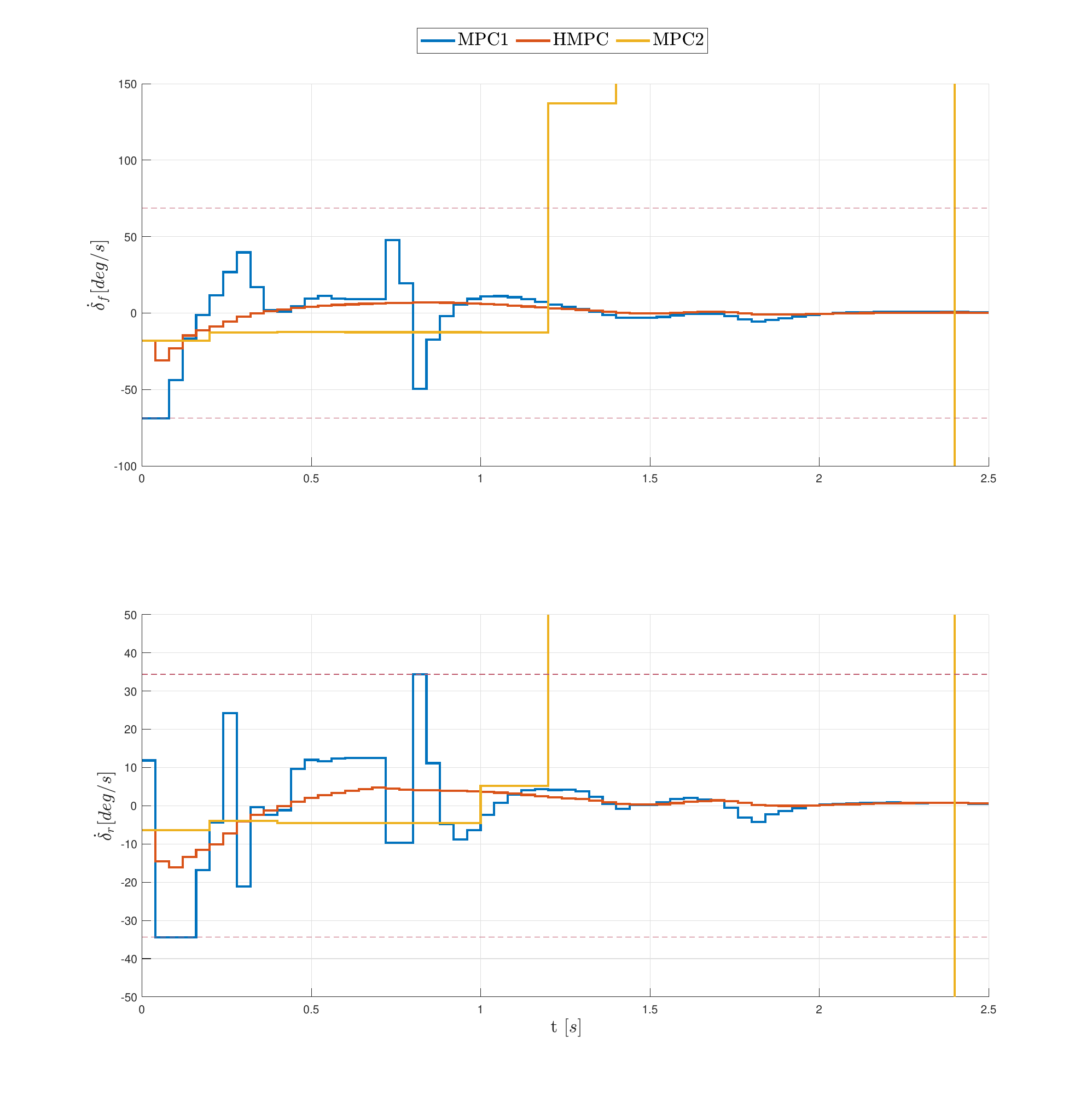}
    \caption{The input trajectory for the nonlinear lane change model. \label{nonlinear_in} }

\end{figure}

\begin{figure}
    \centering
    \includegraphics[trim=0cm 0cm 0cm 0cm, clip=true,scale=0.5]{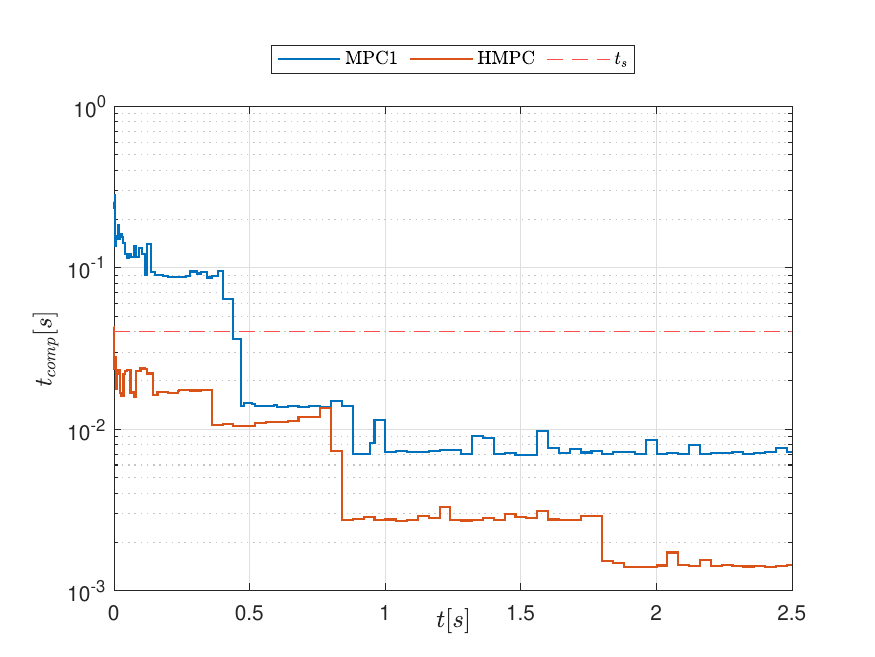}
    \caption{Computation time utilized by MPC1 and HMPC compared to the sampling time $t_s$. The computation time of MPC2 is not included due to the fact that the response is unstable. \label{nonlinear_time} }

\end{figure}

\section{Conclusion and Future Work}
This paper analyzed the stability properties of Hypersampled Model Predictive Control, which is a stratagem for decoupling the sampling time of MPC from the computational complexity of the underlying optimal control problem. Unlike existing MPC stability proofs, the proposed stability analysis does not assume that the system dynamics match the prediction model. This important distinction enables the sampling time and discretization time to be treated as two fully independent quantities. By enabling the control designer to select each time constants based on different considerations, HMPC provides a simple yet effective way to increase the sampling rate of MPC without increasing its computational footprint. Theoretical analysis shows that the HMPC scheme is Input-to-State Stable to input disturbances. Future work will focus on ensuring robust constraint enforcement.


\begin{thebibliography}{10}

\bibitem{mayne2000mpc}
D.~Q. Mayne, J.~B. Rawlings, C.~V. Rao, and P.~O.~M. Scokaert, ``Constrained
  model predictive control: Stability and optimality,'' {\em Automatica},
  vol.~36, no.~6, pp.~789--814, 2000.

\bibitem{tenny2004nonlinear}
M.~J. Tenny, S.~J. Wright, and J.~B. Rawlings, ``Nonlinear model predictive
  control via feasibility-perturbed sequential quadratic programming,'' {\em
  Computational Optimization and Applications}, vol.~28, pp.~87--121, 2004.

\bibitem{diehl2009efficient}
M.~Diehl, H.~J. Ferreau, and N.~Haverbeke, ``Efficient numerical methods for
  nonlinear {MPC} and moving horizon estimation,'' {\em Nonlinear model
  predictive control: towards new challenging applications}, pp.~391--417,
  2009.

\bibitem{quirynen2015autogenerating}
R.~Quirynen, M.~Vukov, M.~Zanon, and M.~Diehl, ``Autogenerating microsecond
  solvers for nonlinear {MPC}: a tutorial using acado integrators,'' {\em
  Optimal Control Applications and Methods}, vol.~36, no.~5, pp.~685--704,
  2015.

\bibitem{fbstab}
D.~Liao-McPherson, M.~Huang, and I.~Kolmanovsky, ``A regularized and smoothed
  {Fischer–Burmeister} method for quadratic programming with applications to
  model predictive control,'' {\em IEEE Transactions on Automatic Control},
  vol.~64, no.~7, pp.~2937--2944, 2018.

\bibitem{MPC4}
J.~Rawlings, D.~Q. Mayne, and M.~Diehl, {\em Model Predictive Control: Theory,
  Computation, and Design (Second Edition)}.
\newblock Madison, WI: {Nob Hill Publishing}, 2017.

\bibitem{magni2004model}
L.~Magni and R.~Scattolini, ``Model predictive control of continuous-time
  nonlinear systems with piecewise constant control,'' {\em IEEE Transactions
  on Automatic Control}, vol.~49, no.~6, pp.~900--906, 2004.

\bibitem{feller2017}
C.~Feller and C.~Ebenbauer, ``Relaxed logarithmic barrier function based model
  predictive control of linear systems,'' {\em IEEE Transactions on Automatic
  Control}, vol.~62, no.~3, pp.~1123--1238, 2017.

\bibitem{DE-MPC}
M.~M. Nicotra, D.~Liao-McPherson, and I.~V. Kolmanovsky, ``Dynamically embedded
  model predictive control,'' in {\em 2018 Annual American Control Conference
  (ACC)}, pp.~4957--4962, 2018.

\bibitem{BLANCHINI19991747}
F.~Blanchini, ``Set invariance in control,'' {\em Automatica}, vol.~35, no.~11,
  pp.~1747--1767, 1999.

\bibitem{HAUSER2002377}
J.~Hauser, ``A projection operator approach to the optimization of trajectory
  functionals,'' {\em IFAC Proceedings Volumes}, vol.~35, no.~1, pp.~377--382,
  2002.

\bibitem{dontchev2019lipschitz}
A.~L. Dontchev, I.~V. Kolmanovsky, M.~I. Krastanov, M.~M. Nicotra, and V.~M.
  Veliov, ``Lipschitz stability in discretized optimal control with application
  to {SQP},'' {\em SIAM Journal on Control and Optimization}, vol.~57, no.~1,
  pp.~468--489, 2019.

\bibitem{teel}
A.~Teel, D.~Nesic, and P.~Kokotovic, ``A note on input-to-state stability of
  sampled-data nonlinear systems,'' in {\em Proceedings of the 37th IEEE
  Conference on Decision and Control}, vol.~3, pp.~2473--2478, 1998.

\bibitem{gilbert1991linear}
E.~G. Gilbert and K.~T. Tan, ``Linear systems with state and control
  constraints: The theory and application of maximal output admissible sets,''
  {\em IEEE Transactions on Automatic control}, vol.~36, no.~9, pp.~1008--1020,
  1991.

\bibitem{TDO_MPC}
D.~Liao-McPherson, M.~M. Nicotra, and I.~V. Kolmanovsky, ``Time-distributed
  optimization for real-time model predictive control: Stability, robustness,
  and constraint satisfaction,'' {\em Automatica}, vol.~117, p.~108973, 2020.

\end{thebibliography}
\end{document}